\documentclass[reqno,11.5pt]{elsarticle}

\usepackage{amssymb}	
\usepackage{tikz}
\usepackage{tikz-cd}

\usepackage{amsmath}
\usepackage{mathrsfs}
\usepackage{amsthm}
\usepackage{verbatim}
\usepackage{graphicx}
\usepackage{color}
\usepackage{float}
\usepackage[colorlinks,linkcolor=blue,anchorcolor=green,citecolor=red]{hyperref}
\usepackage[
labelfont=sf,
hypcap=false,
format=hang,
]{caption}

\usepackage[english]{babel}
\usepackage{mathtools}
\usepackage{hyperref}

\addto\extrasenglish{%
}    

\usepackage{hyperref}

\usepackage[utf8]{inputenc}

\usepackage{cleveref}
\crefname{section}{§}{§§}
\Crefname{section}{§}{§§}
\usepackage{bm}
\usepackage{amsfonts}
\usepackage{amscd}
\usepackage{cases}
\usepackage{xcolor}
\usepackage{subeqnarray}

\definecolor{red}{rgb}{0.8,0,0}
\definecolor{darkorange}{rgb}{1,0.4,0}
\definecolor{lightorange}{rgb}{1,0.6, 0}
\definecolor{yellow}{rgb}{1,0.8, 0}

\newtheorem{theorem}{Theorem}

\newtheorem{definition}{Definition}
\newtheorem{corollary}{Corollary}
\newtheorem{lemma}{Lemma}

\newcommand{\Reals}[1]{{\rm I\! R}^{#1}}

\usepackage{algorithm}
\usepackage[noend]{algpseudocode}

\newcommand\grad{\operatorname{grad}}

\renewcommand\div{\operatorname{div}}

\newcommand\curl{\operatorname{curl}}

\newcommand\bu{\bm{u}}
\newcommand\bv{\bm{v}}

\newcommand\x{\times}

\usepackage{chngpage}

\usepackage{lipsum}

\usepackage{booktabs}

\makeatletter
\@addtoreset{equation}{section}
\makeatother

\usepackage{tgschola}

\usepackage{geometry}
\begin{document}
\title{Helicity-conservative Physics-informed Neural Network Model for Navier-Stokes Equations}

\author[1,2]{Jiwei Jia}
\ead{jiajiwei@jlu.edu.cn}
\author[3]{Young Ju Lee}
\ead{yjlee@txstate.edu}
\author[1]{Ziqian Li}
\ead{zqli21@mails.jlu.edu.cn}
\author[1]{Zheng Lu}
\ead{luzheng21@mails.jlu.edu.cn}
\author[1,2]{Ran Zhang\corref{cor1}}
\ead{zhangran@jlu.edu.cn}

\cortext[cor1]{Corresponding author}
\affiliation[1]{organization={School of Mathematics, Jilin University},
city={Changchun, Jilin},
postcode={130012},
country={China}}

\affiliation[2]{organization={National Applied Mathematical Center (Jilin)},
city={Changchun, Jilin},
postcode={130012},
country={China}}
  
\affiliation[3]{organization={Department of Mathematics, Texas State University},
city={San Marcos, TX},
postcode={78666},
country={USA}}

\begin{abstract}
We design the helicity-conservative physics-informed neural network model for the Navier-Stokes equation in the ideal case. The key is to provide an appropriate PDE model as loss function so that its neural network solutions produce helicity conservation. Physics-informed neural network model is based on the strong form of PDE. We compare the proposed Physics-informed neural network model and a relevant helicity-conservative finite element method. We arrive at the conclusion that the strong form PDE is better suited for conservation issues. We also present  theoretical justifications for helicity conservation as well as supporting numerical calculations. 
\end{abstract}


\maketitle
\section{Introduction}\label{sec:intro} 

Numerical modeling and simulation for the incompressible Navier-Stokes system is critical in a number of applications. Therefore, there have been a lot of efforts in designing numerical methods for solving the incompressible Navier-Stokes equations. It is well-known that the Navier-Stokes system has various conserved quantities and these quantities are important to preserve in numerical calculations to attain stability and meaningful physical solution. Especially, for stability, the energy law and the incompressibility have been proved to be crucial. A focus in our paper is on the helicity of a divergence-free vector field, or fluids helicity, which is conserved in non-dissipative systems, i.e., ideal flows. Helicity is a standard measure for the extent to which the field lines wrap and coil around one another \cite{cantarella1999influence}, thus, fluid helicity is known to be important in the turbulence regime as discussed in, for example, \cite{frisch1975possibility,perez2009role}. Helicity is also known to provide a local lower bound for the energy \cite{arnold1999topological}, i.e., a topological obstruction of energy relaxation. Readers can find more discussions about fluids helicity in the following references and references cited therein  \cite{arnold1999topological,berger1984topological,moffatt1981some,moffatt1992helicity,moffatt2014helicity}. In many numerical algorithms, such conservative quantities are generally, approximated up to a discretization error, rather than exactly conserved. Such approximation errors may pollute the solution behavior and thus it is of great interest to construct numerical methods that precisely preserve the conservative quantities. Due to the importance of preserving the fluids helicity, there have been some existing efforts to design helicity-preserving schemes in literatures. Liu and Wang \cite{liu2004energy} studied helicity-preserving finite difference methods for axisymmetric Navier-Stokes and MHD flows. Rebholz \cite{layton2008helicity,rebholz2007energy} constructed energy- and helicity-preserving finite element methods for the Navier-Stokes equations. Some further discussion on the helicity for studying turbulence model can be found at \cite{olshanskii2010note}. Kraus and Maj \cite{kraus2017variational} studied helicity-preserving schemes for the magnetic hydrodynamics system based on discrete exterior calculus. Recently, helicity preserving finite elements for MHD, which includes fluids helicity have been introduced in \cite{hu2021helicity}. In this paper, we consider to design helicity preserving neural network models for incompressible Navier-Stokes equation. To the best of our knowledge, it is the first attempt to do so. Our finding is that unlike the standard finite difference or finite element schemes, the helicity preserving scheme is more transparent within the physics informed neural network framework \cite{raissi2019physics} to preserve the helicity for the incompressible Navier-Stokes equation.

The goal of this paper is to construct the neural network model that can preserve the fluids helicity. Unlike the standard finite element methods based on the weak formulation of the PDE models, Physics-informed neural networks (PINN) model \cite{raissi2019physics} is based on the strong PDE and thus, conservation can be shown to be made much easier without having to introduce a number of auxiliary variables. This is really the case for conserving the incompressibility. Later \cite{li2020fourier}, \cite{long2018pde} made great progress on deep learning solving PDEs and \cite{hong2021priori} discussed the theory of it. \cite{krishnapriyan2021characterizing} proposed seq2seq strategy which is essential for timing problems. As is well-known the weak formulation seeks the pair of finite elements that is stable as well as is compatible such that the pressure space contains divergence of velocity so that the strong divergence is attained. However, such a construction is extremely difficult in general. Oftentimes, it relies on a discrete differential form point of view such as the finite element exterior calculus, see   \cite{gawlik2019variational,gawlik2020conservative} on structure-preserving discretization for the fluid mechanics with $H(\div)$-conforming velocity. Recently, helicity has been discussed for the magnetic hydrodynamics equation. The flow equation is similar to  similar discretization for the Navier-Stokes equations based on the N\'ed\'elec edge element can be found in \cite{girault1990curl}.  However, helicity-preservation was not addressed there. 

The rest of the paper is organized as follows. In Section \ref{sec:preliminary}, we provide preliminaries, notation and helicity-conservative finite element scheme. In Section \ref{sec:pinn}, we present a PINN-based algorithm that preserves the helicity. In Section \ref{sec:num}, we present numerical results on the convergence and helicity-preserving properties of our algorithms. In Section \ref{sec:con}, we give some concluding remarks.




\section{Governing Equation and Conservative Quantities}\label{sec:preliminary} 

We consider the following system of equations in $\Omega \times (0, \mathbb{T}]$:
\begin{subeqnarray}\label{main:eq2} 
\partial_t \bu - \bu\times \bm{\omega} - R_{e}^{-1}\nabla \times \nabla \times \bu + \nabla\left (\frac{1}{2}|\bu|^{2} + \widetilde{p} \right ) &=& \bm{f},  \slabel{main:eq2m} \\ 
\bm{\omega} &=& \nabla \times \bm{u}, \slabel{main:eq3m} \\ 
\nabla \cdot {\bf{u}} &=& 0. \slabel{main:eq4m}  
\end{subeqnarray}
Here $\partial_t {\bf{u}} = \partial {\bf{u}} / \partial t$, ${\bf{u}}$ and $p$ are the fluid velocity and pressure, respectively. This formulation of the momentum equation \eqref{main:eq2m} is called the Lamb form \cite{lamb1932hydrodynamics}. 

We shall consider the following boundary conditions for \eqref{sec:intro}:  for all {$\bm{x}\in\partial \Omega$} and {$t > 0$},
\begin{eqnarray}\label{NS-bc}
\bu \times \bm{n} = \bm{0}, \quad p := \widetilde{p} + \frac{1}{2} |\bu|^2 = 0.
\end{eqnarray}
where $\bm{n}$ is the unit outer normal vector.
In fact, the conditions for $\bm{u}$ and $\widetilde{p}$ in \eqref{NS-bc} can be seen as a vorticity boundary condition since $\bm{u}\times \bm{n}=0$ implies $(\nabla\times \bm{u})\cdot \bm{n}=0$ on $\partial \Omega$. Girault \cite{girault1990curl} used similar conditions as \eqref{NS-bc} for the Navier-Stokes equations. As we shall see below, this is the natural boundary condition that appears in the helicity conservation.


The initial conditions for the fluid velocity, magnetic field are given for any {$\bm{x}\in\Omega$}
\begin{equation}
\bu(\bm{x},0) =  \bu_{0}(\bm{x}).
\end{equation}
We now review some conserved quantities of \eqref{NS-bc} below. First, we can present the energy conservation without the proof (see \cite{hu2021helicity} for example).  
\begin{theorem}
The diffusion or conservation of energy can be stated as follows.
The NS system \eqref{main:eq2} with the boundary condition \eqref{NS-bc} has the following energy identity: 
 \begin{eqnarray}\label{resistive:energylaw} 
{1 \over 2}\frac{d}{dt} \|\bu\|_0^2  + R_{e}^{-1} \|\nabla \times \bu\|_0^2 
= (\bm{f}, \bu). 
\end{eqnarray}
\end{theorem}
We note that the energy law \eqref{resistive:energylaw} is important since it induces the energy norm stability: 
\begin{eqnarray*}
\max_{0\leq t\leq T} 
    \|\bu\|_0^2 
    + R_{e}^{-1} \int_{0}^{T} \|\nabla \times 
    \bu\|_0^2\,\mathrm{d}\tau 
    \leq  ~ \|\bu_{0}\|_0^{2}  + 
    {R_{e}} \int_0^T\|\bm{f}\|_{-1}^{2}\,\mathrm{d}\tau. 
\end{eqnarray*}
Secondly, the main interest of this paper is the fluids helicity. We begin by defining the helicity
\begin{definition}\label{hel} 
For any divergence-free field $\bm{\xi}$ in $\Omega$, the helicity of $\bm{\xi}$ is defined as 
\begin{equation} 
\mathcal{H}_{\bm{\xi}}:=\int_{\Omega} \bm{\xi}\cdot \bm{\eta}\, dx, 
\end{equation} 
where $\bm{\eta}$ is any potential of $\bm{\xi}$ satisfying $\nabla\times \bm{\eta}=\bm{\xi}$. By an integration by parts, $\mathcal{H}_{\bm{\xi}}$ does not depend on the choice of the potential $\bm{\eta}$ if $\bm{\xi}\cdot\bm{n}=0$ on the boundary $\partial \Omega$, c.f., \cite{arnold1999topological}. 
\end{definition} 
We remark that the fluids helicity describes the linking and knots of the field $\bm{\xi}$. In fluid mechanics, one defines the fluid helicity, denoted by $\mathcal{H}_f$ as follows: 
\begin{equation}
\mathcal{H}_{f} := \int_{\Omega} \bm{u} \cdot \nabla \times \bm{u}\, \mathrm{d}x = \int_{\Omega} \bm{u} \cdot \bm{\omega}\, \mathrm{d}x. 
\end{equation} 
We state the precise form of helicity conservation as follows. 
\begin{theorem}
For the Euler equation, the following identity holds: 
\begin{equation} \label{id:Hm}
\frac{\mathrm{d}}{\mathrm{d}t} \mathcal{H}_f =  2 R_{e}^{-1} \int_{\Omega}\nabla \times \nabla \times \bm{u} \cdot \bm{\omega} \, {d}x + 2 \int_\Omega \bm{f} \cdot \bm{\omega} \, dx. 
\end{equation} 
\end{theorem}
\begin{proof}
We notice that by the integration by parts together with the boundary condition for $\bm{u}$, 
\begin{eqnarray*}
\frac{d}{dt} \int_\Omega \bm{u} \cdot \bm{\omega} \, dx &=&   \int_\Omega \partial_t \bm{u} \cdot \bm{\omega} \, dx +  \int_\Omega\bm{u} \cdot \nabla \times  \partial_t \bm{u} \, dx =  2 \int_\Omega \partial_t \bm{u} \cdot \bm{\omega} \, dx \\ 
&=& 2 \int_\Omega \left ( \bu \times \bm{\omega} + R_{e}^{-1}\nabla \times \nabla \times \bu - \nabla\left (\frac{1}{2}|\bu|^{2} + \widetilde{p} \right ) + \bm{f} \right ) \cdot \bm{\omega} \, dx \\ 
&=& 2 \int_\Omega \left ( R_{e}^{-1}\nabla \times \nabla \times \bu + \bm{f} \right ) \cdot \bm{\omega} \, dx \\ 
\end{eqnarray*}
This completes the proof. 
\end{proof}
This theorem tells that for the ideal case when $R_e = \infty$ and $\bm{f} = \bm{0}$, we obtain the helicity conservation. On the other hand, the natural pollution term for the helicity conservation is the effect of diffusion. Therefore, throughout the discussion that follows, we shall assume $\bm{f} = \bm{0}$. 

\subsection{Finite element preserving the fluids helicity} 
\label{Finite element}

In this section, we shall consider Helicity-conservative finite element methods defined on contractible domains.  One can extend the definition of helicity to nontrivial topology and different space dimensions  \cite[Chapter 3]{arnold1999topological}. We use the standard notation for the inner product and the norm of the
{$L^{2}$} space
$$
(u,v):=\int_{\Omega}u\cdot v \mathrm{d}x,\quad
\|u\|_{0}:=\left(\int_{\Omega} \lvert u\rvert^2 \mathrm{d}x\right)^{1/2}.
$$
 Define the following $H(D,\Omega)$ space with a given linear operator {$D$}:
$$
H(D,\Omega):=\{v\in L^2(\Omega), Dv\in L^2(\Omega)\},  
$$
and
$$
H_0(D,\Omega):=\{v\in H(D, \Omega), t_{D}v=0 \mbox{ on } \partial\Omega\},
$$
where $t_{D}$ is the trace operator:
$$
t_{D}v:=
\left\{
  \begin{array}{cc}
    v, & D=\mathrm{grad},\\
    v\times n, & D=\mathrm{curl},\\
    v\cdot n, & D=\mathrm{div}.
  \end{array}
\right.
$$
We also define:
$$
L^2_0(\Omega):=\left\{v\in L^2(\Omega):  \int_\Omega v=0 \right\}.
$$
By definition, $H_0(\mathrm{grad}, \Omega)$ coincides with $H^1_0(\Omega)$.


The de~Rham complex in three space dimensions with vanishing boundary conditions reads:
 \begin{equation}\label{sequence3-0}
\begin{tikzcd}
0 \arrow{r}  \arrow{r} &H_{0}({\grad}, \Omega)  \arrow{r}{\grad}& {H}_{0}(\curl, \Omega)\arrow{r}{\curl} &{H}_{0}(\mathrm{div}, \Omega)\arrow{r}{\div}& L_{0}^{2}(\Omega) \arrow{r}{} & 0.
 \end{tikzcd}
\end{equation}
The sequence \eqref{sequence3-0} is exact on contractible domains, meaning that  $\mathcal{N}(\curl)=\mathcal{R}(\grad)$ and $\mathcal{N}(\div)=\mathcal{R}(\curl)$, where $\mathcal{N}$ and $\mathcal{R}$ denote the kernel and range of an operator, repectively. 

A main idea of the discrete differential forms \cite{arnold2018finite,Arnold.D;Falk.R;Winther.R.2006a} or the finite element exterior calculus \cite{Bossavit.A.1998a,Hiptmair.R.2002a} is to construct finite elements for the spaces in \eqref{sequence3-0} such that they fit into a discrete sequence
 \begin{equation}\label{sequence3-0-h}
{\begin{tikzcd}
0 \arrow{r}  &H_{0}^{h}({\grad}, \Omega)  \arrow{r}{\grad}& {H}_{0}^{h}(\curl, \Omega)\arrow{r}{\curl} &{H}^{h}_{0}(\mathrm{div}, \Omega)\arrow{r}{\div}& L_{0}^{2, h}(\Omega) \arrow{r}{} & 0.
 \end{tikzcd}} 
\end{equation}


The discrete de~Rham sequences can be  of arbitrary
order \cite{Boffi.D;Brezzi.F;Fortin.M.2013a,
  Arnold.D;Falk.R;Winther.R.2006a}.
 Figure~\ref{deRham} shows the finite elements of the lowest order (Whitney forms).
\begin{figure}[ht!]
\begin{center}
\includegraphics[width=.8in]{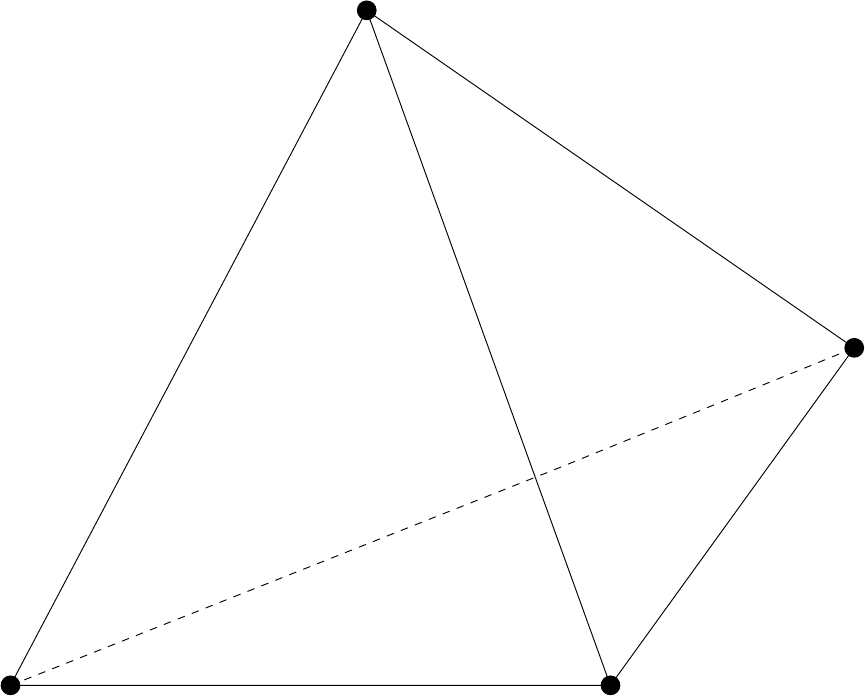}
\includegraphics[width=.8in]{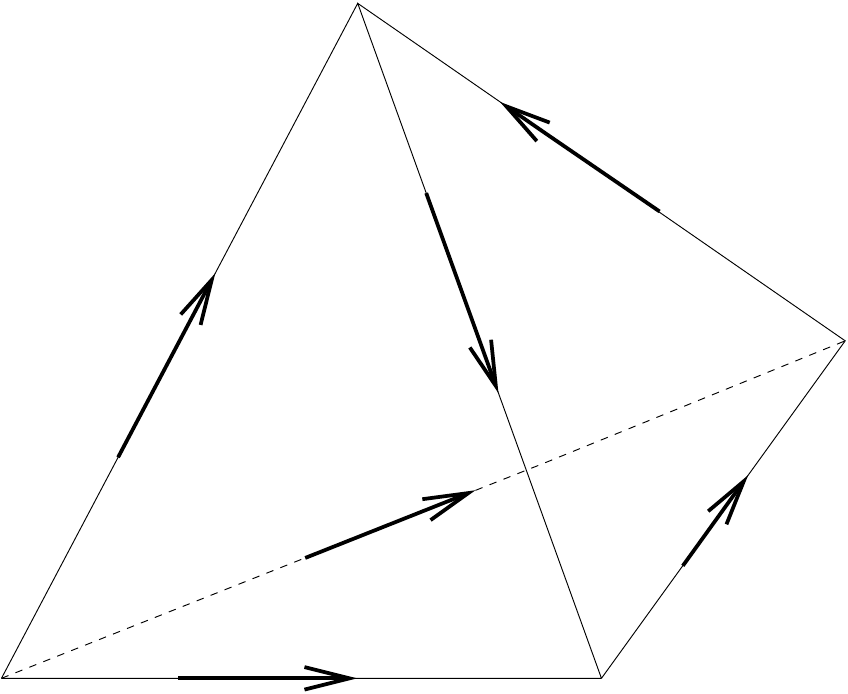}
\includegraphics[width=.8in]{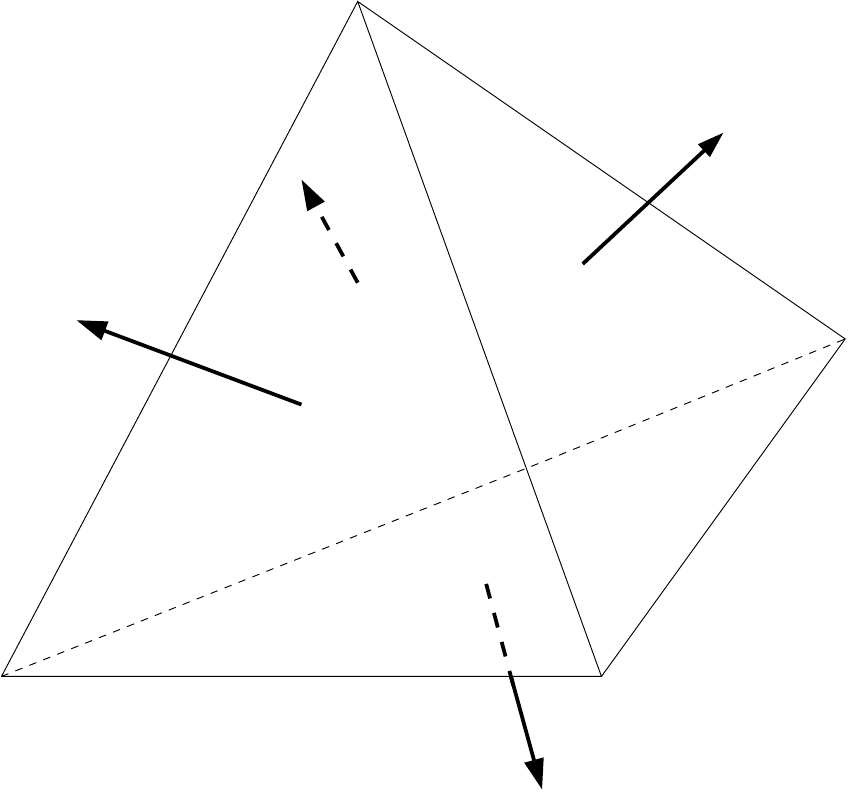}
\includegraphics[width=.8in]{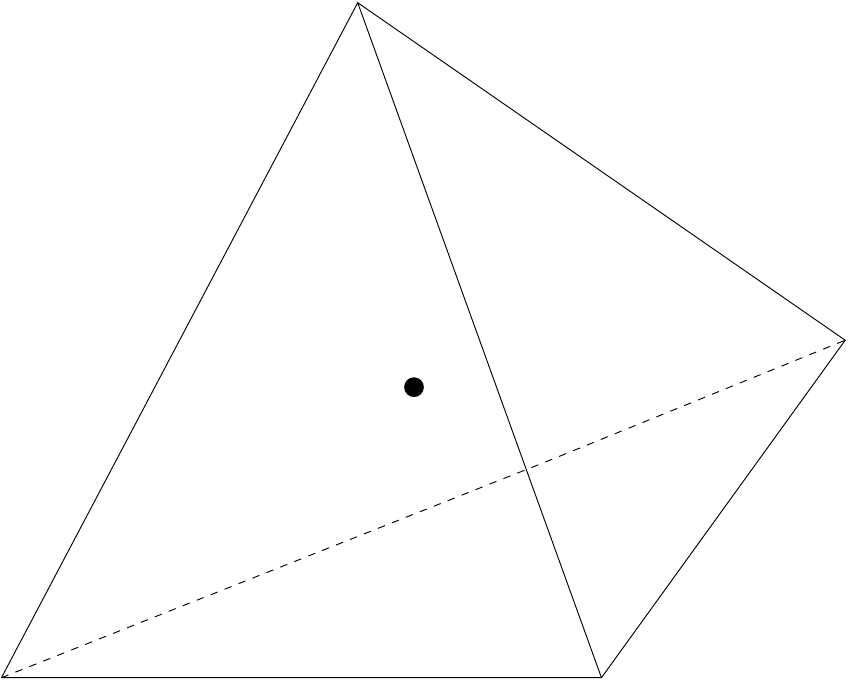}
\caption{DOFs of the finite element de~Rham sequence of lowest order}
\label{deRham}
\end{center}
\end{figure}


Throughout this section, we shall use $\mathbb{Q}_h^{\rm curl}$, the $L^2$ projection onto $H_0^h({\curl},\Omega )$: 
\begin{equation}
\mathbb{Q}^{\rm curl}_h: \left [ L^{2}(\Omega)\right ]^{3}\mapsto H^{h}_{0}({{\curl}},\Omega) 
\end{equation} 
and the operator $\nabla_h \times : H^h_0({\rm div}, \Omega) \mapsto H_0^h({{\curl}}, \Omega)$ defined by the following relation: 
\begin{equation}
(\nabla_h \times \bm{U}, \bm{V}) = (\bm{U}, \nabla \times \bm{V}) \quad \forall (\bm{U},\bm{V}) \in H_0^h({\rm div},\Omega) \times H_0^h({{\curl}},\Omega).  
\end{equation} 

We are now in a position to present the helicity-preserving full discretization for the Navier-Stokes system. Generally, for helicity preserving scheme, mid-point method in time discretization is used \cite{layton2008helicity,rebholz2007energy}. We shall define
\begin{equation}
\bm{X}_{h} =  H^{h}_{0}({{\curl}},\Omega) \times H^{h}_{0}(\curl,\Omega) \times H_0^h(\rm grad).
\end{equation} 
 The scheme can be written as follows: find $(\bu, \bm{\omega},  p) \in \bm{X}_{h}$, 
such that for all $(\bv, \bm{\mu},  q)  \in \bm{X}_{h}$:
\begin{subeqnarray}\label{main:eq3d}  \slabel{eqn1}
\left( D_t \bu, \bv \right) - (\bu \times \bm{\omega}, \bv) + R_e^{-1}(\nabla\times \bu, \nabla\times \bv) + (\nabla p, \bv)  &=& (\bm{0}, \bv), \qquad  \\\slabel{eqn4}
(\bm{\omega}, \bm{\mu}) - (\nabla\times\bu,  \bm{\mu})&=& 0, \\ \slabel{eqn7}
(\bu ,\nabla q) &=& 0, 
\end{subeqnarray} 
where 
\begin{equation} 
D_t \bu := \frac{\bu^{n+1} - \bu^n}{\Delta t},  
\end{equation} 
and all the unknowns without time derivative, i.e, $\bu,  \bm{\omega}, p$  are defined at the midpoint in time interval $[t_n,t_{n+1}]$, e.g.,  
\begin{equation}
\bu = \frac{\bu^{n+1} + \bu^n}{2}.  
\end{equation}  
We note that this weak form can be equivalently formulated as to find $\bm{u}, \bm{\omega}, p \in \bm{X}_h$ such that  
\begin{equation}\label{discretemethodfull} 
D_t \bu - {\mathbb{Q}_h^{\curl}} ( \bu \times {\mathbb{Q}_h^{\curl}} [\nabla \times \bu] ) + \nabla p  + R_{e}^{-1}{\nabla}_{h}\times\nabla\times \bm{u} = \bm{0}.  
\end{equation}
We shall now list some of properties and the helicity conservation holds in a very limited case. We first introduce a simple but important lemma, whose proof can be found at \cite{hu2021helicity}: 
\begin{lemma} 
For any given $\bm{\xi} \in \Reals{3}$, we have 
\begin{eqnarray}
(D_t \bm{\xi}, \bm{\xi}) = \frac{1}{2\Delta t} \left ( \|\bm{\xi}^{n+1}\|^2 - \|\bm{\xi}^n\|^2 \right ). 
\end{eqnarray}
Furthermore, for any pair of vectors $(\bm{\xi}, \bm{\nu}) \in \mathbb{R}^{3} \times \mathbb{R}^{3}$, we have the following identity: 
\begin{equation} 
D_t \int \bm{\xi} \cdot \bm{\nu} = \int D_t \bm{\xi} \cdot \bm{\nu} + \int \bm{\xi} \cdot D_t \bm{\nu},  
\end{equation}  
where 
\begin{equation}
D_t \int \bm{\xi} \cdot \bm{\nu} := \frac{1}{\Delta t} \left (\int \bm{\xi}^{n+1} \cdot \bm{\nu}^{n+1} -  \int \bm{\xi}^{n} \cdot \bm{\nu}^{n} \right ).
\end{equation} 
\end{lemma} 
We now show the energy law for \eqref{main:eq3d}:  
\begin{theorem}
The discrete energy law holds: 
\begin{equation}
 (D_t \bu, \bu)  + R_{e}^{-1} \|\nabla\times \bu\|^2 = 0. 
\end{equation}
\end{theorem}

We shall now discuss the fluid helicity. The following theorems can be similarly stated and proved for  any contractible subdomain $\tilde{\Omega}\subset \Omega$ if the variables satisfy the boundary conditions \eqref{NS-bc}, or even weaker conditions, on the boundary of $\tilde{\Omega}$. Therefore we obtain identities for both local and global helicity. For simplicity of presentation, we focus on the helicity on $\Omega$, i.e., the global helicity, and stick to the boundary conditions \eqref{NS-bc}. 

\begin{theorem} Any solution for \eqref{main:eq3d} satisfies the following identity for the fluid helicity: 
\begin{equation} 
D_t \int_{{\Omega}}\bu\cdot\bm{\omega} = -2 R_{e}^{-1}\int_{{\Omega}}\nabla\times \bm{u}\cdot \nabla\times \bm{\omega}. 
\end{equation} 
\end{theorem}
\begin{proof} 
We note the following discrete weak form in the following formulation: 
\begin{equation}\label{dmfull} 
D_t \bu - {\mathbb{Q}_h^{\curl}} ( \bu \times {\mathbb{Q}_h^{\curl}} [\nabla \times \bu] ) + \nabla p  + R_{e}^{-1}{\nabla}_{h}\times\nabla\times \bm{u}  = \bm{0}.  
\end{equation}
We note that we can define the fluid helicity using either the total integral of 
$\bm{u} \cdot \nabla \times \bm{u}$ or $\bm{u} \cdot \bm{\omega}$, for both of which, we note that the following holds:  
\begin{eqnarray*} 
D_t \int_{ {\Omega}} \bu \cdot \bm{\omega} &=& \int_{ {\Omega}} D_t \bu \cdot \bm{\omega} + \int_{ {\Omega}} \bu \cdot D_t \bm{\omega}= 2 \int_{ {\Omega}} D_t \bu \cdot \bm{\omega} \\
&=& 2 \left [ \int_{{\Omega}} [ \bu \times {\mathbb{Q}}_h^{\curl} (\nabla \times \bu) ] \cdot \bm{\omega}  +\int_{{\Omega}} \nabla p \cdot \bm{\omega} -R_{e}^{-1}\int_{{\Omega}}\nabla\times \bm{u}\cdot \nabla\times \bm{\omega} \right ] \\
&=& 
-2R_{e}^{-1}\int_{{\Omega}}\nabla\times \bm{u}\cdot \nabla\times \bm{\omega}.
\end{eqnarray*} 
This implies the desired result. 
\end{proof} 

We remark that in Girault, the following finite element scheme has been introduced \cite{girault1990curl}: find $(\bu, p) \in H^h_0({\rm curl},\Omega) \times H^h_0({\rm grad})$, 
such that for all $(\bv, \bm{\mu},  q)  \in H^h_0({\rm curl},\Omega) \times H^h_0({\rm grad})$:
\begin{subeqnarray}\label{main:eq4d}  \slabel{eqn41}
\left( D_t \bu, \bv \right) - (\bu \times \nabla \times \bm{u}, \bv) + R_e^{-1}(\nabla\times \bu, \nabla\times \bv) + (\nabla p, \bv)  &=& (\bm{f}, \bv), \qquad  \\\slabel{eqn42}
(\bu ,\nabla q) &=& 0, 
\end{subeqnarray} 
subject to the same boundary condition with the scheme \label{main:ed3d}. 
However, we can see that this scheme does not preserve helicity (see Figure \ref{hahahihi} below). 

\section{Helicity-conservative PINN models}\label{sec:pinn} 

We shall now look at the PINN model, which uses the strong form of PDE as a loss function. We continue to consider the following 3D Navier-Stokes equation with initial condition as given in \eqref{main:eq2}.  On the other hand, for the helicity conservation, we shall assume that the external force is zero. Therefore, we shall assume $\bm{f} = \bm{0}$.  

We note that we are training $p = \frac{1}{2}\left|\bm{u}\right|^2+\tilde{p}$, $\tilde{p}$ is the output $p_{NN}$ of PINN, and $p$ is the modified pressure. Our model minimizes the following functionals for finding $\bm{u}_{NN}$ and $p_{NN}$. Initially, the term of initial condition of $\bm{u}$ and $p$ is:
\begin{equation}
\mathcal{L}_{{\textsf{init}}_{(\bm{u},p)}} = \frac{1}{N_0} \sum_{i=1}^{N_0} \left ( \left \| \bm{u}(t^i, \bm{x}^i) - \bm{u}^i \right\|_0^2 + \left |p(t^i, \bm{x}^i) - p^i \right|^2 \right ), 
\label{Loss_init_NN}
\end{equation}
where $N_0$ is the number of sampling points for initial condition, and $\bm{u}_{NN} = (\bm{u}^i = \{u_1^i, u_2^i, u_3^i\})_{i=1}^{i=N_0}$ and $p_{NN} = (p^i)_{i=1}^{i=N_0}$ denote the initial training data on 
\begin{equation}
\bm{u}(t^i,\bm{x}^i) = \{u_1(t^i,x^i,y^i,z^i), u_2(t^i,x^i,y^i,z^i), u_3(t^i,x^i,y^i,z^i)\}
\end{equation}
and $p(t^i,x^i,y^i,z^i)$ for $i=1,\cdots,N_0$. 
Now, the loss function for boundary conditions can be defined according to \eqref{NS-bc}:
\begin{equation}
\mathcal{L}_{\textsf{bdry}_{(\bm{u},p)}} = \frac{1}{N_b} \sum_{i=1}^{N_b}  \left ( \|\bm{u}^i \times \bm{n} \|_0^2 
+ \left ( \frac{1}{2} \left \|\bm{u}^i \right \|_0 + \tilde{p}^i \right )^2 \right ), 
\label{loss-bound}
\end{equation}
where $N_b$ is the number of sampling points for boundary condition. Next we consider the loss function for the momentum equation as well as the loss function for the divergence free condition for $\bm{u}$, respectively as follows: 
\begin{subeqnarray}
\mathcal{L}_{\textsf{momentum}} &=& \frac{1}{N_f} \sum_{i=1}^{N_m} \left ( \partial_t \bm{u}^i - \bm{u}^i \times \nabla\times \bm{u}^i + R_e^{-1} \nabla \times \nabla\times\bm{u}^i + \nabla p^i \right )^2,
\label{loss_f} \\ 
\mathcal{L}_{\textsf{divergence}} &=& \frac{1}{N_f} \sum_{i=1}^{N_f} |{\rm div} \bm{u}^i |^2,  
\label{loss_div}
\end{subeqnarray}
where $N_f$ is the number of sampling points for the equation. The functional that the PINN model to optimize is given as follows: 
\begin{equation}
\min_{\bm{u}_{NN}, p_{NN}} \mathcal{L}_{\textsf{init}_{(\bm{u},p)}} + \mathcal{L}_{\textsf{bdry}} + \mathcal{L}_{\textsf{momentum}} + \mathcal{L}_{\textsf{divergence}}
\label{loss}
\end{equation}
The original PINN approach trains the NN model to predict the entire space-time at once. In complex cases, this can be more difficult to learn. Seq2seq strategy was proposed in \cite{krishnapriyan2021characterizing}, where the PINN learns to predict the solution at each time step, instead of all times. Note that the only data available of the first sequence is from the PDE itself, i.e., just the initial condition. We take the prediction at $t=$d$\mathbb{T}$ by using the model of the first sequence and use this as the initial condition to make a prediction in the next sequence, and so on. which can be shown in Figure \ref{seq2seq}.

\begin{figure}[H]
\centering
   \includegraphics[width=1\textwidth]{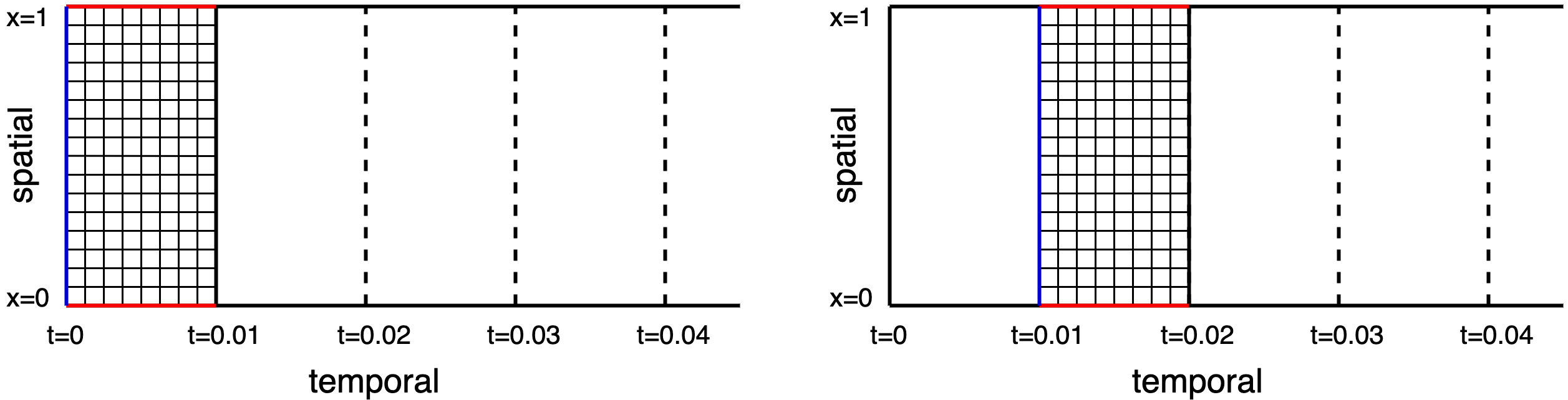}
   \caption{Seq2seq PINN. The blue line is initial condition for each sequence. The red line is boundary condition for each sequence. The domain will be uniformly sectioned. When training of the first sequence finished, the solution at $t=0.01$ will be calculated and used as the initial condition for the next sequence.}
   \label{seq2seq}
\end{figure}
After training our model, we calculate whether the helicity conserves. The helicity is defined as follows. Here we use Gaussian quadrature to replace the integral.Throughout the paper, $Q_\Omega(\cdot)$ denotes the Gaussian quadrature on $\Omega$.
\begin{equation}
\mathcal{H}_f = \int_{\Omega}\bm{u}\cdot \nabla \times \bm{u}\ {\rm{d}} x 
\label{helicity}
\end{equation}
Computationally, we want realize it using the neural network function as follows: 
\begin{equation}
\mathcal{H}^{up_{NN}}_f = \int_{\Omega} \bm{u}_{NN} \cdot \nabla \times \bm{u}_{NN} \ {\rm{d}} x,  
\label{helicityd}
\end{equation}
where $\bm{\omega} = \nabla \times \bm{u}_{NN}$. In conclusion, our helicity-conservative PINN model shows as Algorithm \ref{PINN_algorithm}

\begin{algorithm}
\caption{$up_{NN}$ Network}
\label{PINN_algorithm}
\begin{algorithmic}[1]
	\State {\bf{Input:}} Initial sample points for $\bm{u}$, $p$ and sample points in $\Omega$ for numerical integration and training algorithm, and the number of layers $L$, the final time level $\mathbb{T}$ and the time domain of one sequence d$\mathbb{T}$.   
	\State {\bf{Output:}} $\bm{u}$, $\bm{w}$, $p$ in the sample points in $\Omega$.
	\State Initialize PINN model parameters, weights and bias denoted by  
	$(\bm{W}^{1,(0)},...,\bm{W}^{L,(0)})$ and $(\bm{b}^{1,(0)},...,\bm{b}^{L,(0)})$
	\State Generate initial points and corresponding $\bm{u}$, $\bm{w}$, $p$
    \For{seq=1, $\cdots$ ,$\mathbb{T}$/d$\mathbb{T}$}
    \State Generate sample points in $\Omega$
    \Repeat
    	\State Compute $\bm{u}$ and $p$ using neural network, and then $\bm{\omega} =  \nabla \times \bm{u}$
    	\State Compute loss of initial condition 
    	\State Compute loss of boundary condition 
    	\State Compute loss of equation 
		\State $Loss = \mathcal{L}_{init_{NN}} + \mathcal{L}_{\textsf{bdry}_{(\bm{u},p)}}+\mathcal{L}_{\textsf{momentum}}+\mathcal{L}_{\textsf{divergence}}$
		\State Optimize $Loss$ based on {\bf{Adam}} \Comment{Learning rate of {\bf{Adam}} changes depending on $Loss$}
    \Until{$Loss <$ particular number}
    
    \State Compute helicity $\mathcal{H}^{up_{NN}}_f$ 
    \State Compute $\bm{u}$, $\bm{w}$, $p$ at the end of this sequence, used as initial condition for next sequence
    \EndFor

\end{algorithmic}
\end{algorithm}

\begin{theorem} 
The Algorithm \ref{PINN_algorithm} admits the neural network solution 
$(\bm{u}_{NN}, p_{NN})$ that satifies the following 
identity for the fluid helicity: 
\begin{eqnarray*} 
D_t Q_{\Omega} \left (\bu_{NN} \cdot \nabla \times \bm{u}_{NN} \right) = -2 R_{e}^{-1} Q_{\Omega} \left ( \nabla\times \bm{u}_{NN} \cdot \nabla \times \nabla \times \bm{u}_{NN} \right ) + 2 Q_{\Omega}( \nabla p_{NN} \cdot \bm{\omega}).  
\end{eqnarray*} 
\end{theorem}
\begin{proof} 
We note the following discrete weak form in the following formulation: 
\begin{equation}\label{dmfull} 
D_t \bu - ( \bu \times [\nabla \times \bu] ) + \nabla p  + R_{e}^{-1}\nabla \times\nabla\times \bm{u} = \bm{0} ,\quad \forall (t^i, \bm{x}^i) \in (0, T] \times \Omega.   
\end{equation}
This leads to 
\begin{eqnarray*} 
D_t \bu \cdot \bm{\omega} &=& D_t \bu \cdot \bm{\omega} + \bu \cdot D_t \bm{\omega} =  2 D_t \bu \cdot \bm{\omega} \\
&=& 2 \left [ (\bu \times \nabla \times \bu) \cdot \bm{\omega} +\nabla p \cdot \bm{\omega} - R_{e}^{-1} \nabla \times \bm{u} \cdot \nabla \times \bm{\omega}  \right ] \\
&=& 2 \nabla p \cdot \bm{\omega} - 2R_{e}^{-1} \nabla\times \bm{u}\cdot \nabla\times \bm{\omega}.
\end{eqnarray*} 
By taking numerical quadrature rule to define the fluid helicity, we obtain the desired result and this completes the proof.
\end{proof} 
\begin{corollary}
If the integral $Q_\Omega$ is done accurately and $R_{e} = \infty$, then the helicity from Algorithm \ref{PINN_algorithm} is conserved exactly.  
\begin{eqnarray*} 
D_t \int_\Omega \left (\bu_{NN} \cdot \nabla \times \bm{u}_{NN} \right) \, dx &=& 0.
\end{eqnarray*} 
\end{corollary}
\begin{proof}
We assume that the integral is computed exactly. Then, we have that
\begin{equation}
Q_\Omega (\nabla p \cdot \bm{\omega})  = \int_\Omega \nabla p \cdot \bm{\omega} \, dx = - \int_\Omega p {\rm div} \bm{\omega} \, dx + \int_{\partial \Omega} p \cdot \bm{n} \bm{\omega}\, ds = 0.
\label{divw-proof}
\end{equation} 
The last equality is due to the fact that $\bm{\omega}$ is divergence-free. This completes the proof. 
\end{proof} 

Helicity-conservative PINN model should preserve the initial helicity value in the time evolution. For demonstrating the helicity conservative property of Helicity-conservative PINN model, We describe the algorithm based on  Girault \cite{girault1990curl} as $\omega_{NN}$ Network. In $\omega_{NN}$ Network, the following form of the equation is considered:
\begin{equation}
	\begin{aligned}
		\partial_t\bm{u}_{NN}-\bm{u}_{NN} \times \bm{\omega}_{NN} +R_e^{-1}\nabla\times\nabla\times\bm{u}_{NN}+\nabla p_{NN} &= \bm{0}, \\
		\bm{\omega}_{NN}-\nabla\times\bm{u}_{NN} &= \textbf{0}, 
	\end{aligned}
	\label{3D-NS}
\end{equation}

This means the initial, boundary and divergence-free condition are all the same as Helicity-conservative PINN model, except that in $\omega_{NN}$ Network, $\bm{\omega}$ is directly generated by neural network, rather than calculated by $\bm{u}$. In summary, the difference between Helicity-conservative PINN model and $\omega_{NN}$ network is as follows. First, the loss of initial $\bm{\omega}$ should be computed directly, rather than $\bm{\omega} = \nabla\times\bm{u}$. Second, the loss of $\bm{\omega}$ should be added, which is:
\begin{eqnarray}
\mathcal{L}_{\bm{\omega}} = \frac{1}{N_f} \sum_{i=1}^{N_f} \|\bm{\omega}^i - \nabla \times \bm{u}^i\|_0^2.
\label{loss_omega}
\end{eqnarray}
Third, the loss function for the momentum equation should be:
\begin{eqnarray}
\mathcal{L}_{\textsf{momentum}} &=& \frac{1}{N_f} \sum_{i=1}^{N_m} \left ( \partial_t \bm{u}^i - \bm{u}^i \times \bm{\omega}^i + R_e^{-1} \nabla \times \nabla\times\bm{u}^i + \nabla p^i \right )^2. 
\label{loss_w_momentum} 
\end{eqnarray}
Fourth, the helicity is defined as:
\begin{equation}
\mathcal{H}^{\omega_{NN}}_f = \int_{\Omega} \bm{u}_{NN} \cdot \bm{\omega}_{NN} \ {\rm{d}} x,  
\label{helicity_w}
\end{equation}
In conclusion, the algorithm of $\omega_{NN}$ Network is presented as Algorithm \ref{omega_network}.
\begin{algorithm}
\caption{$\omega_{NN}$ Network}
\label{omega_network}
\begin{algorithmic}[1]
	\State {\bf{Input:}} Initial sample points for $\bm{u}$, $\bm{\omega}$, $p$ and sample points in $\Omega$ for numerical integration and training algorithm, and the number of layers $L$, the final time level $\mathbb{T}$ and the time domain of one sequence d$\mathbb{T}$.   
	\State {\bf{Output:}} $\bm{u}$, $\bm{w}$, $p$ in the sample points in $\Omega$.
	\State Initialize PINN model parameters, weights and bias denoted by  
	$(\bm{W}^{1,(0)},...,\bm{W}^{L,(0)})$ and $(\bm{b}^{1,(0)},...,\bm{b}^{L,(0)})$
	\State Generate initial points and corresponding $\bm{u}$, $\bm{w}$, $p$
    \For{seq=1,$\cdots$,$\mathbb{T}$/d$\mathbb{T}$} 
    \State Generate scattered points in domain
    \Repeat
    	\State Compute $\bm{u}$, $\bm{\omega}$ and $p$ using neural network.
    	\State Compute loss of initial condition 
    	\State Compute loss of $\bm{\omega}$ 
    	\State Compute loss of boundary condition 
    	\State Compute loss of equation 
    	\State $Loss = \mathcal{L}_{init_{NN}}+\mathcal{L}_{\bm{\omega}}+\mathcal{L}_{\textsf{bdry}_{(\bm{u},p)}}+\mathcal{L}_{\textsf{momentum}}+\mathcal{L}_{\textsf{divergence}}$
		\State Optimize $Loss$ based on {\bf{Adam}} \Comment{Learning rate of {\bf{Adam}} changes depending on $Loss$}
    \Until{$Loss <$ particular number}
    
    \State Compute helicity $\mathcal{H}^{\omega_{NN}}_f$ 
    \State Compute $\bm{u}$, $\bm{w}$, $p$ at the end of this sequence, used as initial condition for next sequence
    \EndFor

\end{algorithmic}
\end{algorithm}

\begin{theorem} 
The Algorithm \ref{omega_network} admits the neural network solution 
$(\bm{u}_{NN}, \bm{\omega}_{NN}, p_{NN})$ that satifies the following identity for the fluid helicity: 
\begin{equation} 
D_t Q_{\Omega} (\bu_{NN} \cdot \bm{w}_{NN}) = -2 R_{e}^{-1} Q_{{\Omega}} (\nabla\times \bm{u}_{NN} \cdot \nabla\times \nabla\times \bm{u}_{NN}) + 2 Q_{\Omega}(\nabla p_{NN} \cdot \bm{\omega}_{NN} ).  
\label{divpw-proof}
\end{equation} 
\end{theorem}
\begin{proof} 
Similarly to the Algorithm \ref{PINN_algorithm}, we can have the following discrete identity since neural network functions and their derivatives can represent zero: 
\begin{equation}\label{dmfull} 
D_t \bu - ( \bu \times [\nabla \times \bu] ) + \nabla p  + R_{e}^{-1}\nabla \times\nabla\times \bm{u} = \bm{0}.  
\end{equation}
We note that we can define the fluid helicity using either the total integral of $\bm{u} \cdot \bm{\omega}$ as follows: 
\begin{eqnarray*} 
D_t Q_{ {\Omega}} (\bu \cdot \bm{\omega}) &=& Q_{ {\Omega}} (D_t \bu \cdot \bm{\omega}) + Q_{ {\Omega}} (\bu \cdot D_t \bm{\omega}) = 2 Q_{ {\Omega}} (D_t \bu \cdot \bm{\omega}) \\
&=& 2 \left [ Q_{{\Omega}} \left( [ \bu \times \bm{\omega} ] \cdot \bm{\omega} \right)  + Q_{{\Omega}} \left ( \nabla p \cdot \bm{\omega} \right ) - R_{e}^{-1} Q_{{\Omega}} \left( \nabla\times \bm{u}\cdot \nabla\times \bm{\omega} \right ) \right ] \\
&=& -2R_{e}^{-1} Q_{{\Omega}} (\nabla\times \bm{u}\cdot \nabla\times \bm{\omega})+ 2Q_{\Omega} \left ( \nabla p \cdot \bm{\omega} \right ) .
\end{eqnarray*} 
This implies the desired result. 
\end{proof} 
The difference in the helicity expression from Algorithm \ref{PINN_algorithm} and Algorithm \ref{omega_network} can be found at the term 
\begin{equation}
Q_{\Omega} \left ( \nabla p_{NN} \cdot \nabla \times \bm{u}_{NN} \right ) \quad \mbox{ and } \quad Q_{\Omega} \left ( \nabla p_{NN} \cdot \nabla \times \bm{w}_{NN} \right ). 
\end{equation} 
The crucial point we make is that even if we obtain the best possible optimizer in Algorithm \ref{omega_network}, it is in general impossible to make the loss function for the vorticity zero since $\bm{\omega}_{NN}$ can not represent $\nabla \times \bm{u}_{NN}$. In fact the best possible $\bm{\omega}_{NN}$ is given as follows: 
\begin{equation}
\bm{\omega}_{NN} = Q_{NN} (\nabla \times \bm{u}_{NN}), 
\end{equation} 
where $Q_{NN}$ is the $L^2$ projection onto the space of neural network functions chosen in the model. This in general can not guarantee the divergence free property of $\bm{\omega}_{NN}$ and thus the helicity conservation of Algorithm \ref{omega_network} may not hold in general. Therefore, we arrive at the following similar but different corollary. 
By taking numerical quadrature rule to define the fluid helicity, we obtain the desired result and this completes the proof.
\begin{corollary}
If the integral $Q_\Omega$ is done accurately and $R_{e} = \infty$, then the helicity is not in general conserved exactly for Algorithm \ref{omega_network}.  More precisely, we have that: 
\begin{eqnarray*} 
D_t \int_\Omega \left (\bu_{NN} \cdot \bm{w}_{NN} \right) \, dx = 2 \int_{\Omega} \nabla p_{NN} \cdot \bm{\omega}_{NN} \, dx. 
\end{eqnarray*} 
\end{corollary}

\section{Numerical Experiments}\label{sec:num} 

We now report a couple of numerical tests. One result is the error analysis with analytic solutions. The other is to show that the proposed scheme can conserve the helicity. The numerical experiments are performed on a workstation with 1 10 Core Intel(R) Xeon(R) Silver 4210R CPU, 1 RTX A5000 GPU, 128GB RAM, and a Ubuntu 20.04 operating system that implements Pytorch. Note that we choose Adam as the optimizer and the learning rate is different depending on the loss term \eqref{loss}. In this case, the training process will guarantee both convergence and speed. We use $tanh$ as the activation function. In series of experiments, we choose $h=1/16$ and $\Delta t=1E-3$ when calculating Gaussian quadrature in order to guarantee the fairness of evaluation criterion.

\subsection{Error analysis with analytic solution for $up_{NN}$ Network} 
In this section, we carry out a 3D error test with the following form of solutions on the domain $\Omega = [0,1]^3$. First we generate analytic solutions of the equation. starting at 
\begin{equation}
\widetilde{p} = h(x) h(y) h(z),  
\end{equation} 
where $h(\mu) = (\mu^2 - \mu)^2$. Further, we let  
\begin{equation}
g_1(t) = 4 - 2t, \quad g_2(t) = 1+t \quad \mbox{ and } \quad g_3(t) = 1-t. 
\end{equation}
We now introduce analytic velocity that satisfy the boundary conditions. Namely, 
\begin{eqnarray*}
\bm{u} = \left( \begin{array}{c} -g_1 h'(x) h(y) h(z) 
\\ - g_2 h(x) h'(y) h(z)   \\  -g_3 h(x) h(y) h'(z) \end{array} \right ).
\end{eqnarray*}
With this setting, $\bm{u} \times \bm{n} = 0$, and the modified pressure $p = |\bm{u}|^2/2 + \widetilde{p}$ satisfies the boundary condition.    

We report the numerical error analysis obtained by solving Navier-Stokes equation obtained by setting $R_e = 10^4$ for the model problem. The interval of the mesh $h=1/16$ and the whole time domain $T=1s$. The time domain of each sequence is $0.01s$ and the time interval of each sequence $\Delta t= 1E-3s$. The $L_2$ error is shown in Table \ref{standard_error}.

\begin{table}[H]
\begin{center}
{{\begin{tabular}{|c|c|c|}
\cline{1-3}
 $\|\bm{u} - \bm{u_h}\|_{L_2}$ & $\| {\bm{\omega}} - {\bm{\omega}}_h\|_{L_2}$ &  $\|p - p_h\|_{L_2}$
\\ \cline{1-3} 
 5.029E-04  &  7.584E-04  &    2.369E-04  \\ \cline{1-3} 
\end{tabular}
\caption{Numerical error analysis for solving Navier-Stokes equation. The time level at which the error is computed is $T = 1$ and $\Delta t = 1E-03$. $R_e = 10^4$.}
\label{standard_error} 
}}\end{center}
\end{table}

In order to show the effect of different space intervals and time intervals, two groups of comparison tests have been done. First, we choose different space interval $h$ to observe the performance of our model. Table \ref{error_h} shows the result of different $h$.

\begin{table}[H]
\begin{center}
{{\begin{tabular}{|c|c|c|c|}
\cline{1-4}
$h$  & $\|\bm{u} - \bm{u_h}\|_{L_2}$ & $\| {\bm{\omega}} - {\bm{\omega}}_h\|_{L_2}$ &  $\|p - p_h\|_{L_2}$
\\ \cline{1-4} 
$2^{-2}$ & 8.315E-03  &  4.264E-02  &    1.564E-04 
\\ \cline{1-4} 
$2^{-3}$ & 1.621E-03  &  2.331E-03  &    3.996E-04 
\\ \cline{1-4} 
$2^{-4}$ & 5.029E-04  &  7.584E-04  &    2.369E-04 
\\ \cline{1-4} 
\end{tabular}
\caption{Error analysis of different $h$ while $\Delta t=1E-03$. The time level at which the error is computed is $T = 1$. $R_e = 10^4$.}
\label{error_h} 
}}\end{center}
\end{table}

Then we choose different time interval $\Delta t$ to observe the performance of our model. Table \ref{error_dt} shows the result of different $\Delta t$. We can conclude that the space interval has more influence to the error than the time interval.

\begin{table}[H]
\begin{center}
{{\begin{tabular}{|c|c|c|c|}
\cline{1-4}
$\Delta t$  & $\|\bm{u} - \bm{u_h}\|_{L_2}$ & $\| {\bm{\omega}} - {\bm{\omega}}_h\|_{L_2}$ &  $\|p - p_h\|_{L_2}$
\\ \cline{1-4} 
2E-3 & 5.618E-04  &  7.604E-04  &    1.334e-04 
\\ \cline{1-4} 
1E-3 & 5.029E-04  &  7.584E-04  &    2.369E-04 
\\ \cline{1-4} 
5E-4 & 5.556E-04  &  7.597E-04  &    5.900E-06  
\\ \cline{1-4} 
\end{tabular}
\caption{Error analysis of different $\Delta t$ while $h=1/16$. The time level at which the error is computed is $T = 1$. $R_e = 10^4$.}
\label{error_dt} 
}}\end{center}
\end{table}

\subsection{Tests for Helicity conservation for $up_{NN}$ Network and $\omega_{NN}$ Network} 

In this section, we test our algorithm for conservation of divergence ($Div$), energy ($E_f$) and fluid helicity ($\mathcal{H}_f^{up_{NN}}$ and $\mathcal{H}_f^{w_{NN}}$). Our initial condition is given for $\bm{u}$, $\bm{\omega}$ and $p$ as follows: 

\begin{equation}
	\begin{aligned}
		u_1 &= -\sin(\pi(x-0.5))\cos(\pi(y-0.5))z(z-1), \\
		u_2 &= \cos(\pi(x-0.5))\sin(\pi(y-0.5))z(z-1), \\
		u_3 &= 0, \\
		p &= 0.
	\end{aligned}
	\label{init-cond}
\end{equation}
As depicted in Figure \ref{init} below, we note that the desired boundary conditions are satisfied:  
\begin{equation} 
\bm{u} \times \bm{n} = 0, \mbox{ on } \partial \Omega. 
\end{equation} 

\begin{figure}[H]
\centering
   \includegraphics[width=6cm, height=6cm]{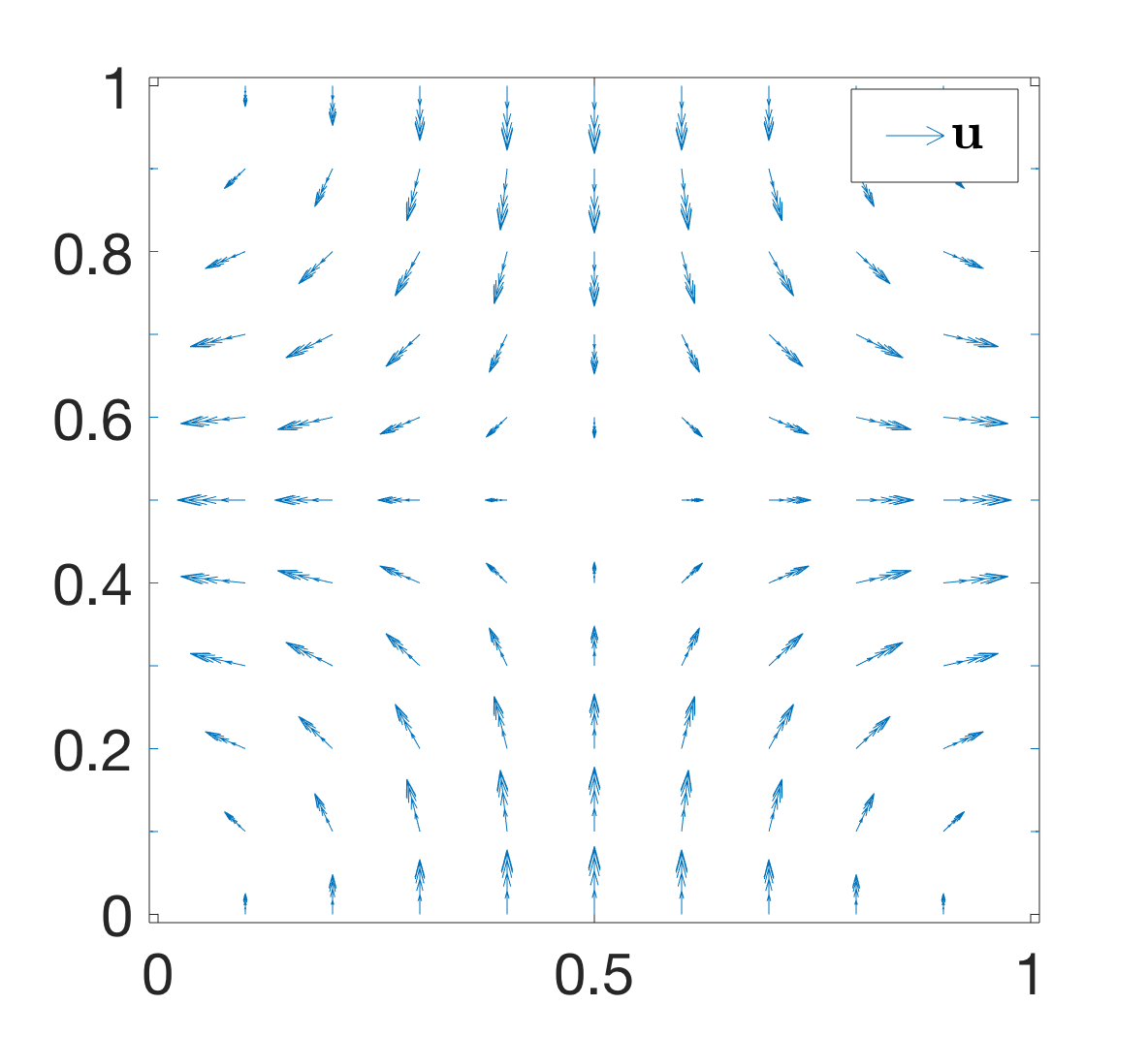}
   \caption{Top view, or projection (onto $xy$ plane) of initial $\bm{u}$}
   \label{init}
\end{figure}

For the pressure $p = \widetilde{p} + \frac{1}{2}|\bm{u}|^2$, we impose the zero boundary condition. Besides, $\bm{f=0}$ in this test. Under this setting, we first show that the divergence conserves well. The divergence is defined as the maximum divergence at certain times, which is shown as Figure \ref{divergence}. The result presents that our model satisfies the divergence-free condition well.

\begin{figure}[H]
\centering
   \includegraphics[width=1\textwidth]{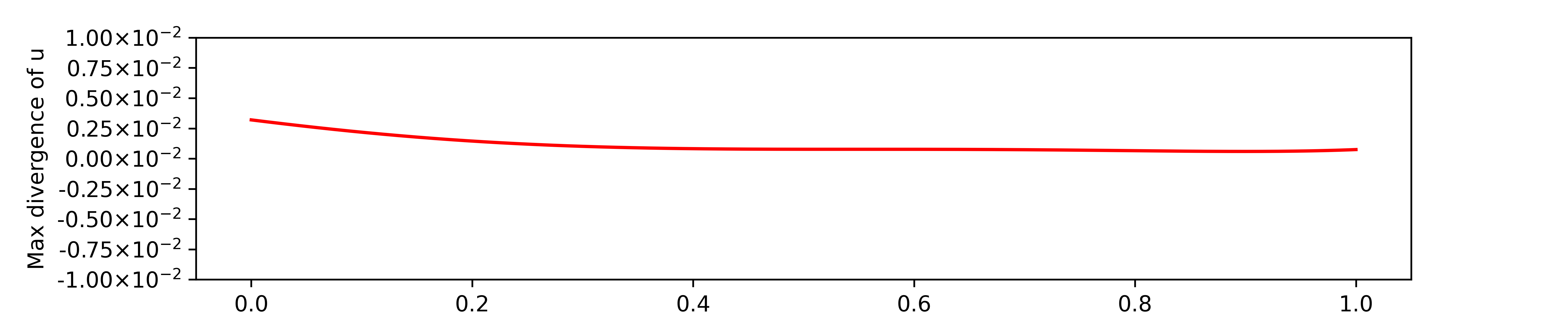}
   \caption{Maximum divergence of $\bm{u}$ for $T=1s$ for $up_{NN}$ Network}
   \label{divergence}
\end{figure}

More importantly, we can show that the energy conserves when $R_e$ approaches $\infty$, which means: 
\begin{eqnarray*}
	\|\bm{u} \|_{L_2} = \sqrt{\int_{\Omega}\left|\bm{u}\right|^2{\rm{d}} x} = 0, \quad \forall t \in [0,\infty).
\end{eqnarray*} 

Energy conservation experiment is shown as Figure \ref{energy}. The figure indicates that $\|\bm{u} \|_{L_2}$ is numerically 0.
\begin{figure}[H]
\centering
   \includegraphics[width=1\textwidth]{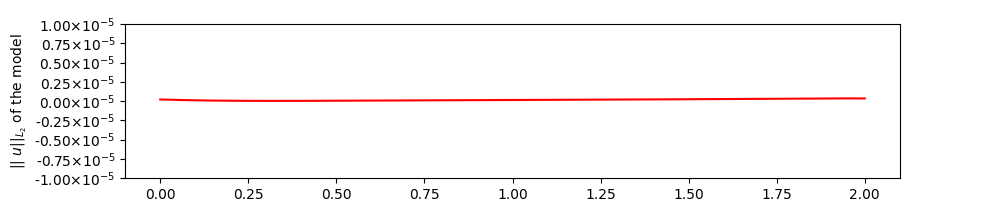}
   \caption{Energy conserves when $T=1s$ for $up_{NN}$ Network}
   \label{energy}
\end{figure}

Additionally, the helicity also conserves:
\begin{eqnarray*}
\mathcal{H}^{up_{NN}}_f = \int_{\Omega} \bm{u}_{NN} \cdot \nabla \times \bm{u}_{NN} \ {\rm{d}} x = 0, \quad \forall t \in [0,\infty).  
\end{eqnarray*}

We call Algorithm \ref{PINN_algorithm} as $up_{NN}$ Network and Algorithm \ref{omega_network} as $\omega_{NN}$ Network. In \ref{divw-proof}, ${\rm div} \bm{\omega}$ is the most significant term for making the formula zero. Divergence of $\omega$ shows as Figure \ref{div_w_png}. This verifies that the divergence of $\omega$ is the main reason fails to conserve the helicity for $\omega_{NN}$ Network.

\begin{figure}[H]
\centering
   \includegraphics[width=1\textwidth]{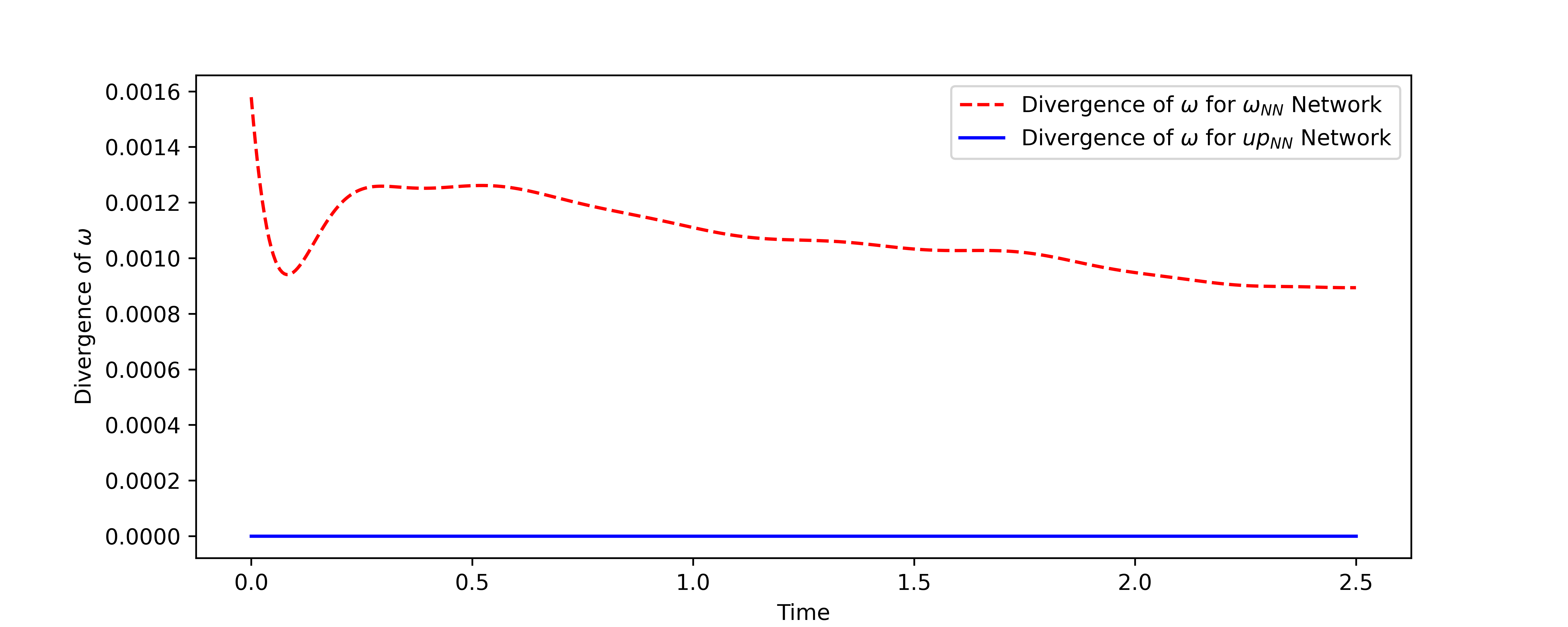}
   \caption{Maximum divergence of $\omega$ for $T=2.5s$}
   \label{div_w_png}
\end{figure}

Further more, \ref{divpw-proof} presents that $\int \nabla p \cdot \bm{\omega}$ is the key that $\omega_{NN}$ Network fails to preserve the helicity. It is shown as Figure \ref{divpw}, which indicates that the $\int \nabla p \cdot \bm{\omega}$ of $up_{NN}$ Network conserves well but $\omega_{NN}$ Network is not.

\begin{figure}[H]
\centering
   \includegraphics[width=1\textwidth]{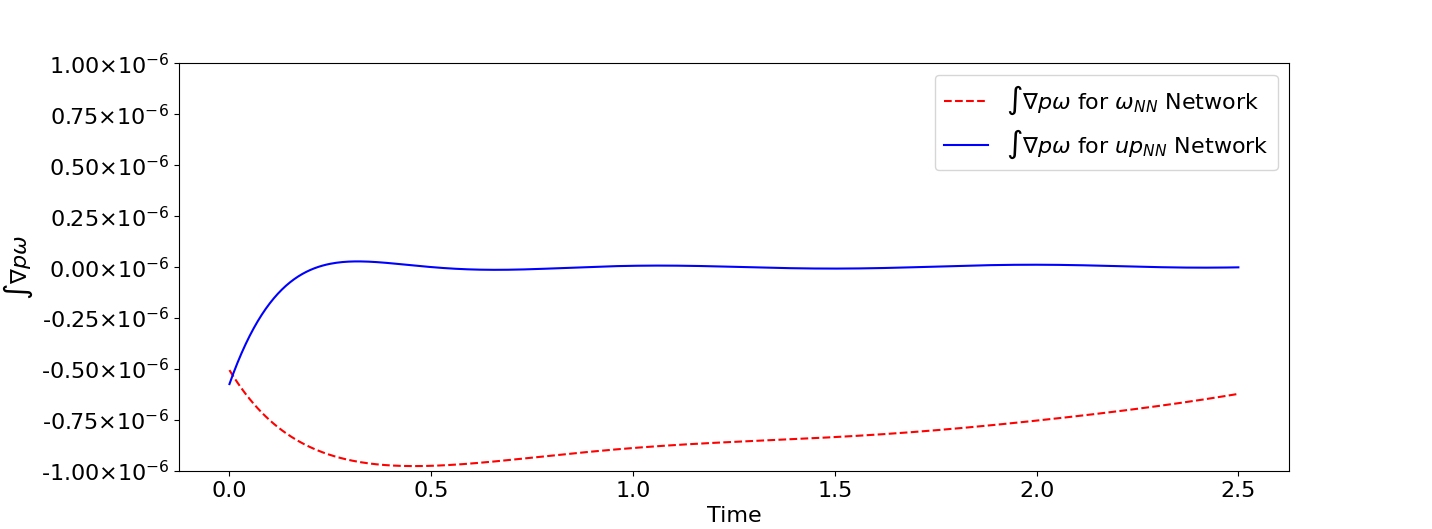}
   \caption{$\int \nabla p \cdot \bm{\omega}$ for $T=2.8s$}
   \label{divpw}
\end{figure}
Experiments show that our scheme preserves the helicity orders of magnitude better with a simple modification in the definition of the vorticity. Figure \ref{helicity} show that the helicity of our model conserves better. Besides, when the interval of mesh becomes larger, the helicity fails to conserve as well as the refining one.
\begin{figure}[H]
\centering
   \includegraphics[width=1\textwidth]{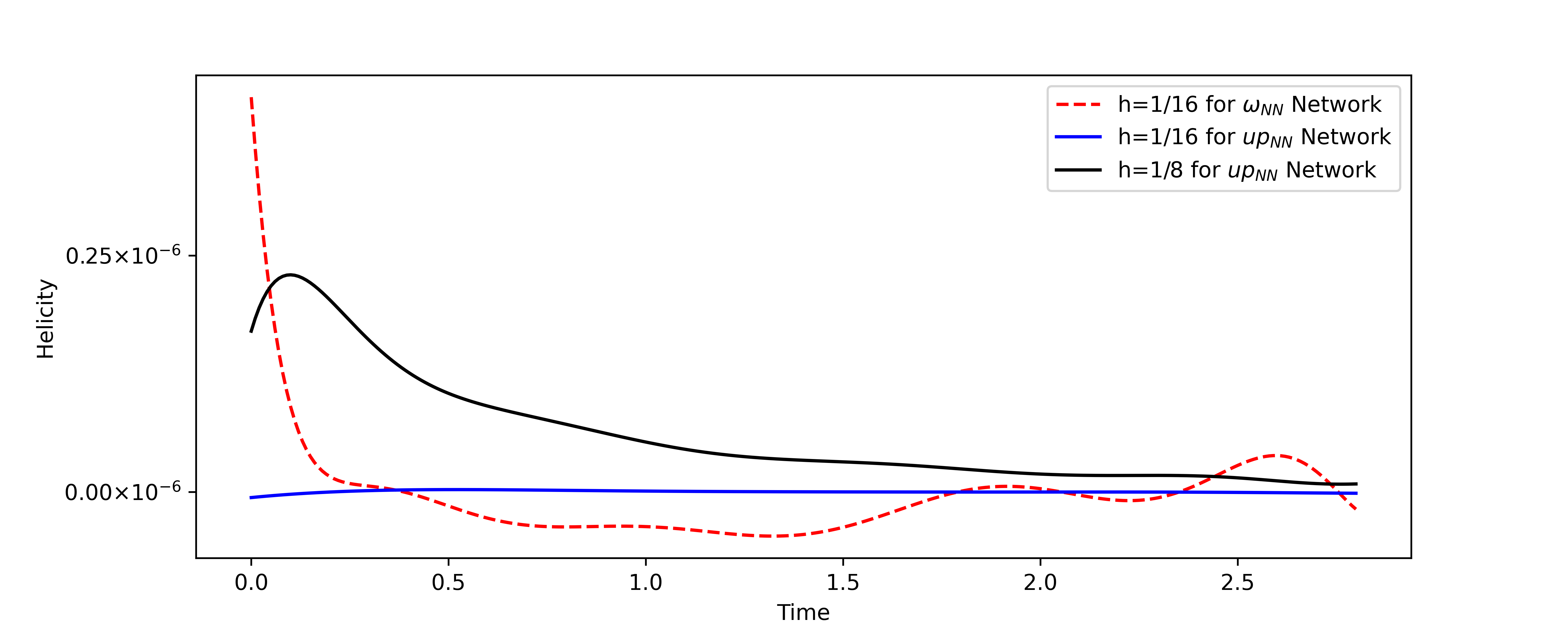}
   \caption{Helicity conservation when $T=2.5s$}
   \label{helicity}
\end{figure}



Figure \ref{hahahihi} shows that the finite element scheme mentioned in \cite{girault1990curl} does not preserve helicity. But our finite element method introduced in section \ref{Finite element} conserves the helicity.

\begin{figure}[h]
\centering 
\includegraphics[width=1\textwidth]{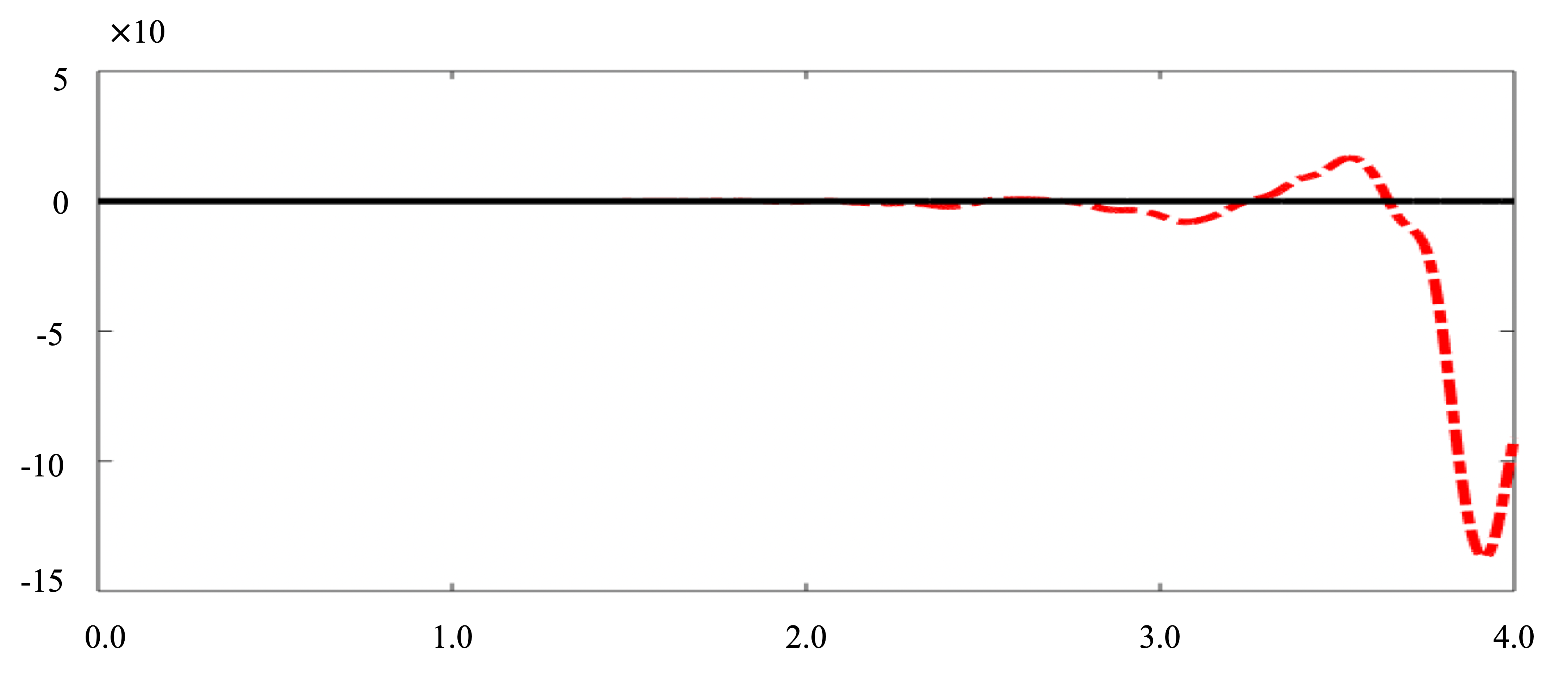}
\caption{$\mathcal{H}_f$ for Girault method(red, dot line) and our finite element method (black, solid line) at $R_e = 10^6$ with $h = 1/16$ and $\Delta t = 1E-03$.}\label{hahahihi} 
\end{figure}



\section{Conclusion}\label{sec:con} 

Neural network is popular and it has a lot of potential. In this paper, we provide a first attempt to use neural network function to preserve helicity of Navier-Stokes equation. Our observation is that PINN model is based on the strong form of PDE, it is easier to demonstrate the conservation property unlike the weak form of PDE. 

\section*{Acknowledgement}

Authors thank Professor Jinchao Xu and Dr. Kaibo Hu for helpful discussion. This work is supported by Graduate Innovation Fund of Jilin University (2023CX183).

\bibliographystyle{plain}      
\bibliography{helicity}{}   

\end{document}